\documentclass[a4paper,UKenglish,cleveref, autoref, thm-restate, numberwithinsect]{lipics-v2021}
\usepackage[T1]{fontenc}
\usepackage{lmodern}

\usepackage{amssymb}
 \usepackage{amsthm}
 \usepackage{xcolor}
\theoremstyle{plain}
\newtheorem{problem}[theorem]{Problem}

\definecolor{gpl_green1}{rgb}{0.18, 0.545, 0.341}

\usepackage[basic]{complexity}
\usepackage{xspace}

\newcommand{\shift}{\mathrm{shift}}
\newcommand{\shifto}{\shift^{\circ}}

\newcommand{\untangling}{\textsc{Circular Untangling}\xspace}
\newcommand{\threepart}{\textsc{3-Partition}\xspace}
\newcommand{\distincticorev}{\textsc{Distinct Increasing Chunk Ordering with Reversals}\xspace}

\newcommand{\ico}{\textsc{Increasing Chunk Ordering}\xspace}

\newcommand{\icorev}{\ico { \sc with Reversals}\xspace}
\newcommand{\icorevS}{\textsc{ICOR}}

\newcommand{\distincticorevS}{\textsc{Dist-ICOR}\xspace}

\newcommand{\remove}[1]{{}}
\newcommand{\new}[1]{\textcolor{black}{#1}}

\hideLIPIcs  
\title{Untangling Circular Drawings: Algorithms and Complexity}



\category{} 

\relatedversion{} 

\supplement{}

\author{Sujoy Bhore}{Indian Institute of Science Education and Research, Bhopal, India.}{sujoy.bhore@gmail.com}{0000-0003-0104-1659}{}

\author{Guangping Li}{Algorithms and Complexity Group, TU Wien, Vienna, Austria}{guangping@ac.tuwien.ac.at}{0000-0002-7966-076X}{}

\author{Martin N\"ollenburg}{Algorithms and Complexity Group, TU Wien, Vienna, Austria}{noellenburg@ac.tuwien.ac.at}{0000-0003-0454-3937}{}

\author{Ignaz Rutter}{University of Passau, Passau, Germany}{Ignaz.Rutter@uni-passau.de}{0000-0002-3794-4406}{}

\author{Hsiang-Yun Wu}{Research Unit of Computer Graphics, TU Wien, Vienna, Austria \\St.~P{\"o}lten University of Applied Sciences, St.~P{\"o}lten, Austria}{hsiang.yun.wu@acm.org}{0000-0003-1028-0010}{}

\authorrunning{S.~Bhore, G.~Li, M.~N\"ollenburg, I.~Rutter, and H.-Y.~Wu}

\Copyright{Sujoy Bhore, Guangping Li, Martin N\"ollenburg, Ignaz Rutter and Hsiang-Yun Wu} 


\category{} 

\relatedversion{} 

\acknowledgements{Guangping Li and Martin N\"ollenburg acknowledge support by the Austrian Science Fund (FWF) under grant P 31119.  
Ignaz Rutter is funded by the Deutsche Forschungsgemeinschaft (DFG, German Research Foundation) -- Ru 1903/3-1.}

\begin{document}
\nolinenumbers
\maketitle


\begin{abstract}
We consider the problem of untangling a given (non-planar) straight-line circular drawing~$\delta_G$ of an outerplanar graph $G=(V,E)$ into a planar straight-line circular drawing by shifting a minimum number of vertices to a new position on the circle. For an outerplanar graph $G$, it is clear that such a crossing-free circular drawing always exists and we define the \emph{circular shifting number} $\shifto(\delta_G)$ as the minimum number of vertices that are required to be shifted in order to resolve all crossings of $\delta_G$. We show that the problem \untangling, asking whether $\shifto(\delta_G) \le K$ for a given integer~$K$, is \NP-complete.
For $n$-vertex outerplanar graphs, we obtain a tight upper bound of $\shifto(\delta_G) \le n - \lfloor\sqrt{n-2}\rfloor -2$.
Moreover, we study the \untangling for \emph{almost-planar} circular drawings, in which a single edge is involved in all the crossings. For this problem, we provide a tight upper bound $\shifto(\delta_G) \le \lfloor \frac{n}{2} \rfloor-1 $ and present a constructive polynomial-time algorithm to compute the circular shifting number of almost-planar drawings.
\end{abstract}




\section{Introduction}
The family of outerplanar graphs, i.e., the graphs that admit a planar drawing where all vertices are incident to the outer face, is an important subclass of planar graphs and exhibits interesting properties in algorithm design, e.g., they have treewidth at most $2$. 
Being defined by the existence of a certain type of drawing, outerplanar graphs are a fundamental topic in the field of graph drawing and information visualization; 
they are relevant to circular graph drawing~\cite{tamassia2013handbook} and book embedding~\cite{DBLP:journals/jct/BernhartK79, BGMN20}.
Several aspects of outerplanar graphs have been studied over the years, e.g., characterization~\cite{Chartrand1967PlanarPG,SYSLO197947, DBLP:journals/siamdm/EllinghamMOT16}, recognition~\cite{10.1007/3-540-17218-1_57, DBLP:conf/approx/BabuKN16}, and drawing~\cite{DBLP:journals/corr/abs-2006-06951, DBLP:journals/tcs/LazardLL19,Six2006AFA}.
Moreover, outerplanar graphs and their drawings have been applied to various scientific fields, e.g., network routing~\cite{DBLP:journals/algorithmica/Frederickson96},  VLSI design~\cite{chung1987embedding}, and biological data modeling and visualization~\cite{krzywinski2009circos, wu:2019:sm}.

In this paper, we study the untangling problem for non-planar circular drawings of outerplanar graphs, i.e., we are interested in restoring the planarity property of a straight-line circular drawing with a minimum number of vertex shift. 
Similar untangling concepts have been used previously for eliminating edge crossings in non-planar drawings of planar graphs~\cite{Goaoc_2009}. 
More precisely, let $G=(V,E)$ be an $n$-vertex outerplanar graph and let $\delta_G$ be  an outerplanar drawing of $G$, which can be described combinatorially as the (cyclic) order 
$\sigma = (v_1, v_2, \ldots, v_n)$ of $V$ when traversing vertices on the boundary of the outer face counterclockwise.
This order $\sigma$ corresponds to a planar circular drawing by mapping each vertex $v_i \in V$ to the point $p_i$ on the unit circle $\mathcal{O}$ with polar coordinate 
$p_i = (1,\frac{i}{n} \cdot 2\pi)$ and drawing each edge $(v_i,v_j) \in E$ as the straight-line segment between its endpoints $p_i$ and $p_j$; see Figure~\ref{fig:morphing}. 

\begin{figure}[tb]
    \centering
	    \begin{subfigure}[t]{0.48\textwidth}
	        \centering
	            \includegraphics[page=1]{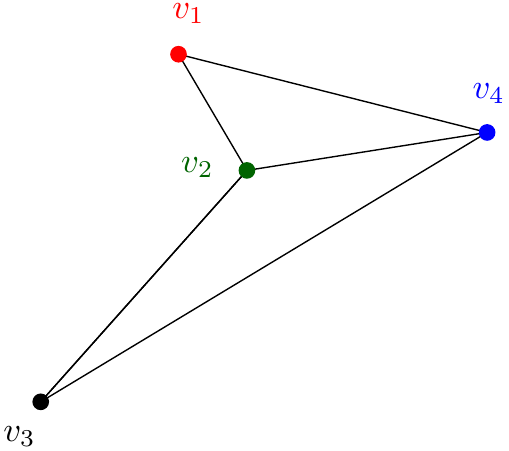}
	\caption{}
	\label{subfig:morphing_a}
	    \end{subfigure}
		\hfill
    \begin{subfigure}[t]{0.48\textwidth}
        \centering
    \includegraphics[page=2]{figs/morphing.pdf}
    \caption{}
        \label{subfig:morphing_b}
    \end{subfigure}%
        \caption{Morphing an outerplanar drawing (a) into a circular drawing (b).}
        \label{fig:morphing}
\end{figure}
We note that it is sufficient to consider circular drawings since any outerplanar drawing can be transformed into an equivalent circular drawing by morphing the boundary of the outer face to $\mathcal{O}$ and then redrawing the edges as straight segments.

Our untangling problem is further motivated by the problem of maintaining an outerplanar drawing of a \emph{dynamic} outerplanar graph, which is subject to edge or vertex insertions and deletions, while maximizing the visual \emph{stability} of the drawing~\cite{DBLP:journals/isci/LinLY11,mels-lam-95}, i.e., the number of vertices that can remain in their current position. Such problems of maintaining drawings with specific properties for dynamic graphs have been studied before~\cite{cohen1995dynamic, diehl2002graphs, beck2014state, BBCCI21}, but not for the outerplanarity property.

\subparagraph*{Related Work.}
The notion of untangling is often used in the literature for a crossing elimination procedure that makes a non-planar drawing of a planar graph crossing-free; 
see~\cite{Cibulka_2009,Kang_2011, polygon02, 10.1007/978-3-642-25870-1_27}.
Given a straight-line drawing $\delta_G$ of a planar graph $G$, the problem to decide whether it is possible to untangle $\delta_G$ 
by moving at most $K$ vertices, is known to be \NP-hard~\cite{Goaoc_2009,VERBITSKY2008294}. 
Lower bounds on the number of vertices that can remain fixed in an untangling process have also been studied~\cite{DBLP:journals/dcg/BoseDHLMW09,DBLP:conf/gd/CanoTU11,Goaoc_2009}. 
On the one hand, Bose et al.~\cite{DBLP:journals/dcg/BoseDHLMW09} proved that 
$\Omega(n^{1/4})$ vertices can remain fixed when untangling a drawing. 
on the other hand, Cano et al.~\cite{DBLP:conf/gd/CanoTU11} gave a family of drawings, where at most $O(n^{0.4948})$ vertices can remain fixed during an untangling process.
Goaoc et al.~\cite{Goaoc_2009} proposed an algorithm, which allows at least $\sqrt{(\log n-1)/\log\, \log\,n}$ vertices to remain fixed when untangling a drawing. 
\remove{If the graph is outerplanar, their algorithm can eliminate all edge crossings while keeping at least $\sqrt{n/2}$ vertices fixed.}
Given an arbitrary drawing of an $n$-vertex outerplanar graph, all edge crossings can be eliminated while keeping at least $\sqrt{n/2}$ vertices fixed~\cite{10.1007/978-3-642-25870-1_27, Goaoc_2009}, whereas there exists a drawing $\delta_G$ of an $n$-vertex outerplanar graph $G$ such that at most $\sqrt{n-1}+1$ vertices can stay fixed when untangling $\delta_G$~\cite{Goaoc_2009}.
Kraaijer et al.~\cite{KraaijerKMR19} proposed several variants on untangling moves such as swapping the locations of two adjacent vertices or rotating an edge over $90$ degrees.
They showed that it is \NP-complete to decide if a drawing can be untangled by swapping. 
They also proved that minimizing the number of swaps needed to untangle an embedded tree is \NP-complete.

Note that the untangled drawings in these previous works are planar but not necessary outerplanar.
In this paper, we study untanglings to obtain an outerplanar circular drawing from a non-outerplanar circular drawing by vertex moves.

\subparagraph*{Preliminaries and Problem Definition.}\label{sec:preliminaries}
Given a graph $G=(V, E)$, we say two vertices are \emph{2-connected} if they are connected by two internally vertex-disjoint paths. 
A 2-connected component of $G$ is a maximal set of pairwise 2-connected vertices. Two subsets $A,B \subseteq V$ are \emph{adjacent} if there is an edge $ab \in E$ with $a \in A$ and~$b \in B$.
A \textit{bridge} (resp.
\textit{cut-vertex}) of $G$ is an edge (resp. vertex) whose deletion increases the number of connected components of $G$.

A drawing of a graph is \emph{planar} if it has no crossings, it is \emph{almost-planar} if there is a single edge that is involved in all crossings, and it is \emph{outerplanar} if it is planar and all vertices are incident to the outer face.
A graph $G=(V, E)$ is \emph{outerplanar} if it admits an outerplanar drawing.
In addition, a drawing where the vertices lie on a circle and the edges are drawn as straight-line segments is called a \emph{circular drawing}.
Every outerplanar graph $G$ admits a planar circular drawing, as one can start with an arbitrary outerplanar drawing $\delta_G$ of $G$ and transform the outer face of $\delta_G$ to a circle~\cite{tamassia2013handbook}.  In this paper, we exclusively work with circular drawings of outerplanar graphs; we thus simply refer to them as drawings.

Given a non-planar circular drawing $\delta_G$ of an $n$-vertex outerplanar graph~$G$ where vertices lie on the unit circle $\mathcal{O}$, we can transform the drawing $\delta_G$ to a planar circular drawing by moving the vertices on the circle $\mathcal{O}$.
Formally, given a circular drawing $\delta_G$, a vertex move operation (or shift) changes the position of one vertex in $\delta_G$ to another position on the circle $\mathcal{O}$~\cite{Goaoc_2009}.
We call a sequence of moving operations that results in a planar circular drawing an \emph{untangling} of $\delta_{G}$. 
We say an untangling is \emph{minimum} if the number of vertex moves of this untangling is the minimum over all valid untanglings of~$\delta_G$. 
We define the \emph{circular shifting number} $\shifto(\delta_G)$ of a circular drawing $\delta_{G}$ as the number of vertex moves in a minimum untangling of~$\delta_G$.
In what follows, we define the problems formally.

\begin{problem}[\underline{\textsc{Circular Untangling (CU)}}] Given a circular drawing $\delta_G$ of an outerplanar graph $G$ and an integer $K$, decide if $\shifto( \delta_G) \le K$.
\end{problem}

\begin{problem}[\underline{\textsc{Minimum Circular Untangling (MinCU)}}] Given a circular drawing $\delta_G$ of an outerplanar graph $G$, find an untangling of $\delta_G$ with $\shifto( \delta_G)$ vertex moves.
\end{problem}

\subparagraph*{Contributions.} 
We show that {\sc Circular Untangling} is \NP-complete in Section~\ref{sec:hardness}.
Then, in Section~\ref{sec: general-bound}, we provide a tight upper bound of the circular shifting number.
Next, we consider almost-planar drawings.
In this case, we provide a tight upper bound on the circular shifting number in Section~\ref{sec:bound} and design a quadratic-time algorithm to compute a circular untangling with the minimum number of vertex moves in Section~\ref{sec: restricted untangling}.
\clearpage
\section{Complexity of Circular Untangling}
\label{sec:hardness}

In this section, we prove the following theorem.

\begin{theorem}
    \label{thm:untangling-is-hard}
    {\sc Circular Untangling} is \NP-complete.
\end{theorem}

The \NP-hardness follows by a reduction from the well-known \NP-complete problem \threepart~\cite{DBLP:books/fm/GareyJ79}. 
However, we do not give a direct reduction but rather work via an intermediate problem, called \distincticorev that concerns increasing subsequences.  
A \emph{chunk} $C$ is a sequence $C = (c_i)_{i=1}^n$ of positive integers.  For a chunk~$C$, we define~$C^1=C$, and we denote its reversal by $C^{-1}$.
We introduce the following problem.

\begin{problem}[\underline{\icorev(\icorevS)}] Given a set $\mathcal C = \{C_1,\ldots,C_\ell\}$ of $\ell$ \emph{chunks} and a positive integer $M$, the question is to determine whether a permutation~$\pi$ of~$\{1,\ldots,\ell\}$ and a function~$\varepsilon \colon \{1,\ldots,\ell\} \to \{-1,1\}$ exist such that the concatenation~$C_{\pi(1)}^{\varepsilon(1)} C_{\pi(2)}^{\varepsilon(2)},\ldots, C_{\pi(\ell)}^{\varepsilon(n)}$ contains a strictly increasing subsequence of length~$M$.
\end{problem}

This problem also comes in a \emph{distinct} variant, denoted \distincticorevS, where all integers in all input chunks are required to be 
distinct.  
We first show that \distincticorevS is \NP-complete and then reduce it to \untangling.
Since we feel that \distincticorevS may serve as a useful starting point for future reductions, we explicitly state our intermediate result.
\begin{theorem}
    \label{thm:disticor-is-hard}
    \distincticorev is \NP-complete.
\end{theorem}

\subsection{Proof of Theorem~\ref{thm:disticor-is-hard}}
\label{subsec: reduction1}

Observe that \distincticorevS lies in \NP, since we can non-deterministically guess an ordering of chunks and whether each of them is reversed or not.  Then the existence of an increasing chain of a specific length can be checked in polynomial time.  The remainder of this section is devoted to showing \NP-hardness by giving a reduction from \threepart.

The input to the \threepart problem consists of a multiset~$A=\{a_1,\ldots,a_{3m}\}$ of $3m$ positive integers and a positive integer $K$ such that~$\frac{K}{4} < a_i < \frac{K}{2}$ for $i=1,\ldots,3m$. 
The question is whether~$A$ can be partitioned into $m$ disjoint triplets $T_1,\ldots,T_m$ such that~$\sum_{a \in T_j} a = K$ for all $j=1,\ldots,m$. It is well-known that \threepart is strongly \NP-complete, i.e., the problem is \NP-complete even if the integers in $A$ and~$K$ are polynomially bounded in~$m$~\cite{DBLP:books/fm/GareyJ79}.
 
Let~$I=(A,K)$ with~$A=\{a_1,\dots,a_{3m}\}$ be an instance of \threepart.
 We assume that each number in $A$ is a multiple of $3m$, 
otherwise, we can multiply each element in $A$ and $K$ by $3m$.
We now construct an equivalent instance $I'=(C,M)$ of \distincticorevS in polynomial time.
\subparagraph*{Construction.}
We create for each element~$a_i$ a corresponding chunk~$C_i$ as follows. 
For two integers $a, l$, we denote the consecutive integer sequence $(a, a+1, \ldots, a+l-1)$ as the  \emph{incremental sequence} of length $l$ starting at~$a$.
We say that a sequence of integers \emph{crosses} an integer $c$ if it contains both a number that is at most~$c$ and a number that is at least~$c+1$.
Let $X = 3mK$.  
We take all the incremental sequences of length~$a_i+X$ starting at $\alpha \cdot (K+3X) + \beta \cdot X +\gamma$ for $\alpha \in \{0,\cdots m-1\}, \beta \in \{0, 1, 2\}$ and $\gamma \in \{1,2,\cdots,K-a_i \}$. Note that there are at most~$X$ such sequences and no such sequence crosses a multiple of $K+3X$.
To construct the chunk  $C_i$, we first build a chunk $C_i'$ with possibly repeating numbers as follows.
The chunk $C_i'$ is formed by concatenating these incremental sequences in decreasing order of their starting number; see Figure~\ref{fig:chunk construction}.  Observe that, in the figure, a strictly increasing subsequence corresponds to an x-monotone chain of points with positive slopes, whereas a non-increasing subsequence corresponds to an x-monotone chain of points with non-positive slopes.
\begin{figure}[hptb]
	\centering
\includegraphics[ page=1]{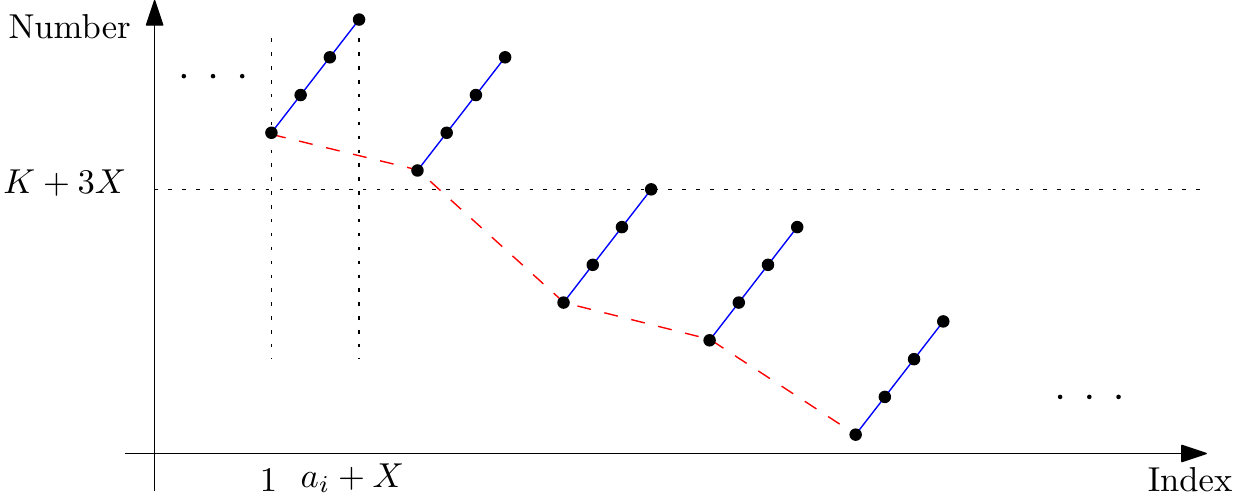}
\caption{Construction of chunk $C_i$ as a concatenation of incremental sequences of length $a_i + X$ in the decreasing order of their first number.
\new{Each blue solid path corresponds to an incremental sequence. 
The red dashed path, which connects the first numbers of incremental sequences, slopes downward.}
}    
\label{fig:chunk construction}
\end{figure}


To make the elements distinct, we introduce strings of numbers, called \emph{words}, which we order lexicographically.
We take the concatenation $C$ of chunks $C_1', C_2', \cdots, C_{3m}'$, then replace the number $a$ at the $i$-th position 
by the word $(a, |C|-i)$ for each position~$i$.  The chunks~$C_1,\dots,C_{3m}$ are obtained by cutting this modified sequence of words in such a way that~$|C_i| = |C_i'|$ for $i=1,\dots,3m$.
At the end of the construction, each word is replaced by its rank in a lexicographically increasing ordering of all words that occur in the instances.
We obtain an instance~$I'=(\mathcal{C},M)$ of~\icorev by setting~$\mathcal{C}=\{C_1,\dots,C_{3m}\}$ and $M := m(K+3X)$.
For simplicity, we use the construction with words in the following. 

For a sequence of words with two entries, we call the sequence obtained by keeping only the first entry of each word, its \emph{projection}.
Note that the projection of~$C_i$ is~$C_i'$.

\begin{lemma}
The chunks~$C_1,\dots,C_{3m}$ have the following properties.
 \begin{enumerate}[(i)]
 \item\label{prop:siss} For every strictly increasing subsequence of~$C_i$, its projection is a strictly increasing sequence.
  \item\label{prop:no-cross} No projection of a strictly increasing subsequence of~$C_i$ crosses a multiple of~$K+3X$.
 \item\label{prop:existence} For $\alpha \in \{0,\cdots m-1\}, \beta \in \{0, 1, 2\}$ and $\gamma \in \{1,2,\cdots,K-a_i \}$, there exists a subsequence of $C_i$ whose projection is the incremental sequence of length~$a_i+X$ starting at $\alpha \cdot (K+3X) + \beta \cdot X +\gamma$.
 \item\label{prop:maxI}Every strictly increasing subsequence of~$C_i$ has length at most~$a_i+X$.
 \item\label{prop:maxD}Every strictly increasing subsequence of~$C_i^{-1}$ has length at most~$X$.
 \end{enumerate}
 \end{lemma}
 \begin{proof}
Since the sequence obtained by keeping the second entry of each word of $C_i$ is  strictly decreasing, we get Property~(\ref{prop:siss}). 
Property~(\ref{prop:no-cross}) and Property~(\ref{prop:existence}) follow directly from the construction of $C_i'$.  

To see the Property~(\ref{prop:maxI}), consider an arbitrary strictly increasing subsequence $s$ of $C_i'$.
Recall that $C_i'$  is the concatenation of incremental sequences of length $a_i +X$ in decreasing order of their starting number.
 Given an index $j \in \{1,\cdots, a_i+X\}$, we claim that $s$ contains the $j$-th element of at most one incremental subsequence of $C_i'$.
 Otherwise, suppose there are two incremental subsequences of $C_i'$ whose $j$-th element is in $s$, then these two numbers are ordered decreasingly, contradicting the assumption that $s$ is strictly increasing. 
 Consequently, the length of strictly increasing subsequence of $C_i'$ is bounded by $a_i+X$.
 
 For Property~(\ref{prop:maxD}), consider a strictly increasing subsequence of $C_i^{-1}$.  It corresponds to a strictly decreasing subsequence $s$ of $C_i$, and its projection is a non-increasing subsequence~$s'$ of $C_i'$.
Note that $s'$ contains at most one element of each incremental sequence of $C_i'$, and $C_i'$ is the concatenation of at most $X$ incremental sequences.  Therefore, the length of $s'$ is at most~$X$.
 \end{proof}






\begin{lemma}
\label{lem:reduction1}
  $I'$ is a yes-instance of~\distincticorevS if and only if~$I$ is a yes-instance of~\threepart.
\end{lemma}
\begin{proof}

Assume there is a partition of the elements of~$A$ into $m$ triplets, each of which sums to~$K$.
We arbitrarily order these triples, and within each triplet, we order the elements according to their index. 
This defines a total ordering on the elements, and therefore on the chunks.    
Let $T_i = \{a_x,a_y,a_z\}$ with~$x < y < z$ be the $i$th triplet and let~$C_x,C_y,C_z$ be the corresponding chunks.  
By Property~(\ref{prop:existence}), $C_x$, $C_y$, and $C_z$ contain, respectively, three subsequences whose projections are the incremental sequences of  length~$a_x + X$, $a_y+ X$, and $a_z+ X$ starting at~$(i-1)(K+3X)+1, (i-1)(K+3X)+X + a_x+1$, and~$(i-1)(K+3X)+2X + a_x+a_y+1$.
Concatenating these subsequences for all chunks hence gives an increasing subsequence whose projection the sequence $1,\cdots,m(K+3X)$.
    
Conversely, assume that there is a chunk ordering that contains a strictly increasing subsequence $S$ of length  $m(K+3X)$. 
By Property~(\ref{prop:maxI}) and Property~(\ref{prop:maxD}), each chunk~$C_i$ or its reversal can contribute a subsequence of at most~$a_i+X$ elements, therefore each chunk~$C_i$ or its reversal must contribute an increasing subsequence of length~$a_i+X$.
Moreover, reversing $C_i$ only provides a shorter increasing subsequence than $a_i+X$, thus no $C_i$ is reversed.
We cut the sequence $S$ into $m$ consecutive sequences $S_1, S_2 \ldots S_m$, called \emph{partition cells} of $S$,  such that the projection of $S_i$ consists of numbers in $\{(i-1)(K+3X)+1, \cdots, i(K+3X)\}$.
By Property~(\ref{prop:no-cross}), the projection of every strictly increasing subsequence inside a chunk does not cross a multiple of~$K+3X$, thus each chunk contributes to exactly one partition cell.  We claim the following:

\begin{claim}
Each partition cell has length $K+3X$.
\end{claim}

We first show how the proof of the lemma can be derived from the claim.  Since the length of each cell is $K+3X$, exactly three chunks contribute to each cell.  Each such triplet of chunks then corresponds to a triplet of $A$ whose sum is $K$.  Together, these triplets define a solution of the instance~$I$ of \threepart.

It remains to prove the claim.  Consider a partition cell $S_i$ consisting of numbers from $n$ chunks.  Then~$S_i$ is the concatenation of subsequences~$S_{i,1},S_{i,2},\dots,S_{i,n}$, $n \le 3m$, each of which is contributed by a different chunk. Since the projection of $S_i$ is a non-decreasing sequence consisting of numbers in $\{(i-1)(K+3X)+1, \cdots, i(K+3X)\}$ and by Property~(\ref{prop:siss}), the projection of each $S_{i,j}$ is a strictly increasing sequence, it follows that non-strict increases of~$S_i$ can only occur when moving from~$S_{i,j}$ to~$S_{i,j+1}$ for some~$j$.
Thus, $|S_i| < K + 3X + n \leq K + 3X + 3m$.  

Note that~$X,K$ and~$|S_i|$ are all multiples of $3m$.  For $X$, this is by definition, for $K$, it follows from the fact that each element of $A$ is a multiple of~$3m$, and for $|S_i|$ recall that each chunk~$C_j$ that contributes a nonempty subsequence of~$S_i$ contributes a sequence of length~$X+a_j$.  Therefore~$|S_i| < K+3X+3m$ implies~$|S_i| \le K+3X$.  Suppose there exists a partition cell $S_j$ with $|S_j|<K+3X$, then $|S| < m(K+3X)$, which contradicts our assumption of $|S| = m(K+3X)$.  Hence~$|S_i| = K+3X$ as claimed.
\end{proof}

\subsection{Proof of Theorem~\ref{thm:untangling-is-hard}}
\label{subsec: reduction2}

It is readily seen that \untangling lies in \NP.  So it remains to describe the reduction from \distincticorevS.
Let~$I=(\mathcal{C}, M)$ be an instance of \distincticorevS with chunks~$C_1,\ldots,C_\ell$.  
By replacing each number with its rank among all occurring numbers, we may assume without loss of generality, that the numbers in the sequence are~$1,\ldots,\sum_{i=1}^\ell |C_i|=:L$.

We construct an instance~$I' =(\delta_G, K)$ of \untangling as follows; see Figure~\ref{fig:reduction3_a}.
We create vertices~$v_1,\ldots,v_L$ and an additional vertex~$v_0$.  
For each chunk~$C_i$, we create a cycle~$K_i$ that starts at $v_0$, visits the vertices that correspond to the elements of $C_i$ in the given order, and then returns to~$v_0$.
That is, $G$ consists of $\ell$ cycles that are joined by the cut-vertex~$v_0$.
The drawing~$\delta_G$ is obtained by placing the vertices in the clockwise order~$\sigma_G = v_0,v_1,v_2,\ldots,v_L$ on~$\mathcal O$.  Finally, we set~$K := L-M$.  Clearly,~$I'$ can be constructed from~$I$ in polynomial time.  It remains to prove the following.

\begin{figure}[ht]
    \centering
	    \begin{subfigure}[t]{0.48\textwidth}
	        \centering
\includegraphics[page=3]{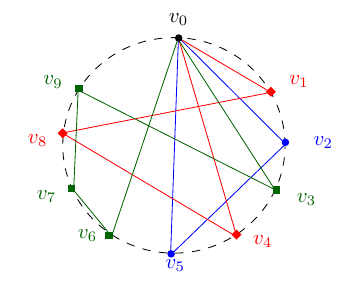}
	\caption{The straight-line circular drawing $\delta_G$ constructed from a \distincticorevS instance with chunks $C_1 = (2,5)$, $C_2 =(1,8,4)$, and $C_3 = (6,7,9,3)$.}
	\label{fig:reduction3_a}
	    \end{subfigure}
		\hfill
    \begin{subfigure}[t]{0.48\textwidth}
        \centering
\includegraphics[page=4]{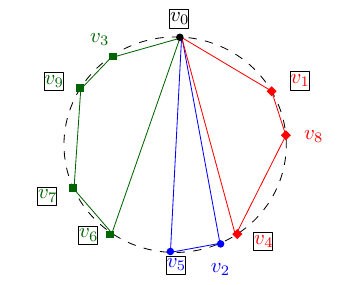}
    \caption{An example drawing obtained by applying an optimum untangling on $\delta_G$. 
Fixed vertices are marked in $\openbox$.  }
        \label{fig:reduction3_b}
    \end{subfigure}%
        \caption{The reduction from \distincticorevS to \textsc {Circular Untangling}}
		\label{fig:reduction3}
\end{figure}

\begin{lemma}
\label{lem:reduction}
  $I$ is a yes-instance of~\distincticorevS if and only if~$I'$ is a yes-instance of~\untangling.
\end{lemma}

\begin{proof}
    Observe that, since in~$\delta_G$ the vertices are ordered clockwise according to their numbering, the problem of untangling with at most~$L-M$ vertex moves is equivalent to finding a planar circular drawing of~$G$ whose clockwise ordering contains an increasing subsequence of at least~$M$ vertices, which can then be kept fixed; see Figure~\ref{fig:reduction3_b}.
    
    Since all the cycles of $G$ are joined at the vertex $v_0$,  the vertices of each cycle~$K_i$ are consecutive in every planar circular drawing of~$G$, and the order of its vertices is the order along~$K_i$, i.e., it is fixed up to reversal.  Hence the choice of a circular drawing whose clockwise ordering contains an increasing subsequence of at least~$M$ vertices directly corresponds to a permutation and reversal of the chunks~$C_i$.
\end{proof}

\section{A Tight Upper Bound of the Circular Shifting Number}
\label{sec: general-bound}

In this section, we investigate an upper bound of the circular shifting number and prove the following theorem.
\begin{theorem}
\label{th:general tight upper bound}
For every drawing $\delta_G$ of an $n$-vertex outerplanar graph $G$, the circular shifting number satisfies $\shifto(\delta_G) \le n - \lfloor\sqrt{n-2}\rfloor -2$, and this bound is tight. 

\end{theorem}

To prove the upper bound, we present an untangling that fixes at least $\lfloor\sqrt{n-2}\rfloor +2$ vertices in the following.
Let~$G=(V,E)$ be an $n$-vertex outerplanar graph with a circular drawing $\delta_G$ of $G$.  Let~$\delta^U_G$ be a planar circular drawing of~$G$.
We number the vertices of~$G$ as~$v_1,\dots,v_n$ in clockwise order according to their occurrence in~$\delta^U_G$.
Now consider untangling $\delta_G$ by moving vertices such that the vertices are ordered as~$v_1,\dots,v_n$ clockwise or counterclockwise.
Doing this with a minimum number of vertex moves is equivalent to finding a longest increasing or decreasing subsequence of the ordering of the vertices in $\delta_G$, which can be fixed during the transformation. 
The claimed bound follows from the Erd\H os-Szekeres Theorem for cyclic permutations.

\begin{theorem}[Erdős–Szekeres theorem for cyclic permutation~\cite{10.2140/involve.2019.12.351}]
\label{th: cyclic-est}
For any two integers $s, r$, any cyclic sequence of $n \geq sr + 2$ distinct real numbers
has an increasing cyclic subsequence of $s + 2$ terms or a decreasing cyclic subsequence of $r + 2$ terms, and this bound is tight.
\end{theorem}

Moreover, observe that for cycles, the planar drawing is unique up to reversal, and therefore untangling a drawing of a cycle with a minimum number of moves is equivalent to determining a longest increasing or decreasing subsequence in the fixed cyclic ordering determined by the cycle.  Hence, a tight example can be obtained from a tight example for the above theorem.

\section{A Tight Upper Bound for Almost-Planar Drawings}
\label{sec:bound}
Let $G=(V,E)$ be an outerplanar graph and let $\delta_G$ be an almost-planar circular drawing of $G$.  
In this section, we present an untangling for such almost-planar circular drawings that provides a tight upper bound of $\lfloor\frac{n}{2}\rfloor-1$ on $\shifto(\delta_G)$.

\begin{theorem}
\label{th:tight upper bound}
For every almost-planar drawing $\delta_G$ of an $n$-vertex outerplanar graph $G$ the circular shifting number satisfies $\shifto(\delta_G) \le \lfloor\frac{n}{2}\rfloor-1$, and this bound is tight. 
\end{theorem}
\begin{figure}[htb]
	\centering
\includegraphics[page=7]{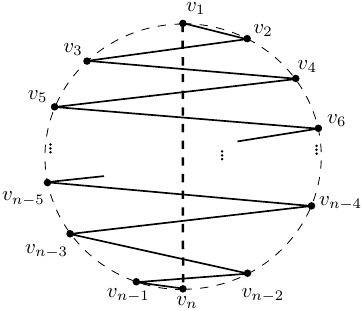}
\caption{An almost-planar drawing $\delta_G$ with $\shifto(\delta_G) = \frac{n}{2}-1$.}    
\label{fig:sect2-3}
\end{figure}
To see that the bound is tight, let $n \ge 4$ be an even number and let $G$ be the cycle on vertices $v_1,\ldots,v_n, v_1$ (in this order) and let~$\delta_G$ be a drawing with the clockwise order $v_2,\ldots, v_{2i}\ldots, v_n, v_{n-1},\ldots,v_{2i+1},\ldots,v_1$; see Figure~\ref{fig:sect2-3}.

We claim that $\shifto(\delta_G) \ge \frac{n}{2}-1$. 
Clearly, the clockwise circular ordering of its vertices in a crossing-free circle drawing is either $v_1,v_2,\ldots,v_n$ or its reversal.
Assume that we turn it to the clockwise ordering $v_1,v_2,\ldots,v_n$; the other case is symmetric.
In $\delta_G$, the $\frac{n}{2}$ odd-index vertices $v_1,\ldots,v_{2i+1}\ldots, v_{n-1}$ and $v_n$ are ordered counterclockwise. 
To reach a clockwise ordering, we need to move all but two of these vertices.
Thus, at least $\frac{n}{2}-1$ vertices in total are required to move.

The remainder of this section is devoted to proving the upper bound.  
Let $e=uv$ be the edge of $\delta_G$ that contains all the crossings, and let $G' = G-e$ and $\delta_{G'}$ be the circular drawing of $G'$ by removing the edge $e$ from $\delta_G$. 
The edge $uv$ partitions the vertices in $V \setminus \{u,v\}$ into the sets $L$ and $R$ that lie on the left and right side of the edge $uv$ (directed from $u$ to $v$).

\begin{theorem}
\label{th:one_side_move}
Let $\delta_G$ be an almost-planar drawing of an outerplanar graph $G$. A planar circular drawing of $G$ can be obtained by moving only vertices of $L$ or only vertices of $R$ to the other side in $\delta_G$ and fixing all the remaining vertices.
The untangling moves only $\min\{|L|,|R|\}$ vertices and can be computed in linear time.
\end{theorem}
This immediately implies the upper bound from Theorem~\ref{th:tight upper bound}, since $|L\cup R| = n-2$, and therefore $\min\{|L|,|R|\} \le  \lfloor\frac{n}{2}\rfloor-1$.
To prove Theorem~\ref{th:one_side_move}, we distinguish different cases based on the connectivity of $u$ and $v$ in $G'$.

\medskip
\noindent\textbf{Case 1: $u,v$ are not connected in $G'$.} Consider a connected component $C$ of $G'$ that contains vertices from $L$ and from $R$.

\begin{proposition}
\label{prop:thcc1}
Suppose $u,v$ are not connected in $G'$. Let $C$ be a connected component of $G'$ that contains vertices from $L$ and from $R$.  It is possible to obtain a new almost-planar drawing $\delta'_G$ of $G$ from~$\delta_G$ by moving only the vertices of $C \cap L$ (resp.\ $C \cap R$) such that $C$ lies entirely on the right (resp.\ left) side of $uv$.
\end{proposition}
\begin{proof}
Since $u,v$ are not connected in $G'$, $C$ contains at most one of $u,v$.
Without loss of generality, we assume that $v \notin C$; see Figure~\ref{fig:sect2-1_a}. 
Let $v'$ be the first clockwise vertex after $v$ that lies in $C$.  
Let $\delta'_G$ be the drawing obtained from $\delta_G$ by moving the vertices of $C \cap L$ clockwise just before $v'$ without changing their clockwise ordering.  
Observe that this removes all crossings of $e$ with $C$. 
The choice of $v'$ ensures that no edge of $C$ alternates with an edge whose endpoints lie in $V\setminus C$.
Finally, the vertices of $C$ maintain their clockwise order.  
This shows that no new crossings are introduced, and the crossings between $e$ and $C$ are removed.
\end{proof}
By applying Proposition~\ref{prop:thcc1} for each connected component of $G'$ that contains vertices from $L$ and from $R$, we obtain a planar circular drawing of $G$.

\begin{figure}[tb]
    \centering
	    \begin{subfigure}[t]{0.48\textwidth}
	        \centering
	            \includegraphics[page=1]{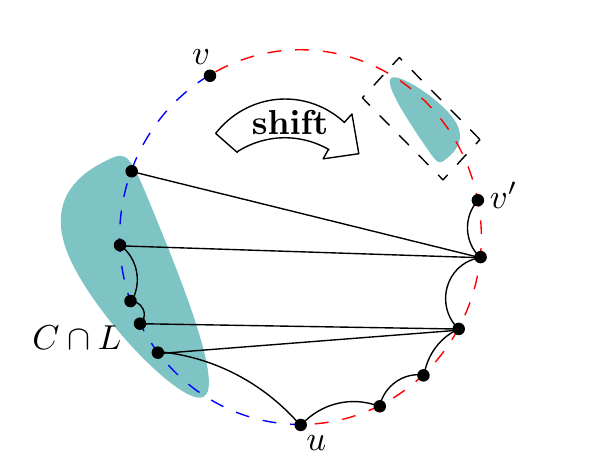}
	\caption{Case 1}
\label{fig:sect2-1_a}
	    \end{subfigure}
		\hfill
    \begin{subfigure}[t]{0.48\textwidth}
        \centering
    \includegraphics[page=2]{figs/section2.pdf}
    \caption{Case 2.2 (non-connecting component)}
\label{fig:sect2-1_b}
    \end{subfigure}%
        \caption{Moving a left component, keeping/reversing the clockwise ordering of its vertices.}
\label{fig:sect2-1}
\end{figure}

\medskip
\noindent\textbf{Case 2: $u,v$ are connected in $G'$.}
Let $C$ be the connected component in $G'$ that contains both vertices $u$ and $v$.
Note that if $C'$ is another connected component of $G'$, then it must lie entirely to the left or entirely to the right of edge $e$. 
Here, we ignore such components as they never need to be moved. 
We may hence assume that $G'$ is connected.

\smallskip
\noindent\textbf{Case 2.1: $u,v$ are 2-connected in $G'$.} We claim that in this case $\delta_G$ is already planar.

\begin{proposition}
\label{prop:Biconected}
  If $u$ and $v$ are 2-connected in $G'$, then $\delta_G$ is planar.
\end{proposition}
\begin{proof}
If vertices $u,v \in V$ are 2-connected in $G'$, then $G'$ contains a cycle $C$ that includes both $u$ and $v$. 
In $\delta_{G'}$, this cycle is drawn as a closed curve.  
Any edge that intersects the interior region of this closed curve therefore has both endpoints on $C$.  
If there exists an edge $e'=xy$ that intersects $e=uv$, then contracting the four subpaths of $C$ connecting each of $\{x,y\}$ to each of $\{u,v\}$ yields a $K_4$-minor in $G$, which contradicts the outerplanarity of $G$.
\end{proof}

\smallskip
\noindent\textbf{Case 2.2: $u,v$ are connected but not 2-connected in $G'$.}
In this case $G'$ contains at least one cut-vertex that separates $u$ and $v$.
Notice that each path from $u$ to $v$ visits all such cut-vertices between $u$ and $v$ in the same order. 
Let $f$ and $l$ be the first and the last cut-vertex on any $uv$-path.
Additionally, add $u$ to the set of $L,R$ that contains $f$ and likewise add $v$ to the set of $L,R$ that contains $l$.
Let $X$ denote the set of edges of $G'$ that have one endpoint in $L$ and the other in $R$.  
Each connected component of $G'-X$ is either a subset of $L$ or a subset of $R$, which are called \emph{left} and \emph{right components}, respectively. 
We call a component of $G'-X$ \emph{connecting} if it contains either $u$ or $v$, or removing it from $G'$ disconnects $u$ and $v$.  
For a left component $C_L$ and a right component $C_R$, we denote by $E(C_L,C_R)$ the set of edges of $G'$ that connect a vertex from $C_L$ to a vertex in $C_R$. 
We can observe that since $G'$ is connected, for any edge that connects a left and a right component, at least one of the components must be connecting. 
We use the following observation.
\begin{observation}
\label{obs:no-skip-path}
If $P$ is an $xy$-path in a left (right) component $C$, then it contains all vertices of $C$ that are adjacent to a vertex of a right (left) component and lie between $x$ and $y$ on the left (right) side.
\end{observation}

 \begin{figure}[tb]
    \centering
\includegraphics[page=1]{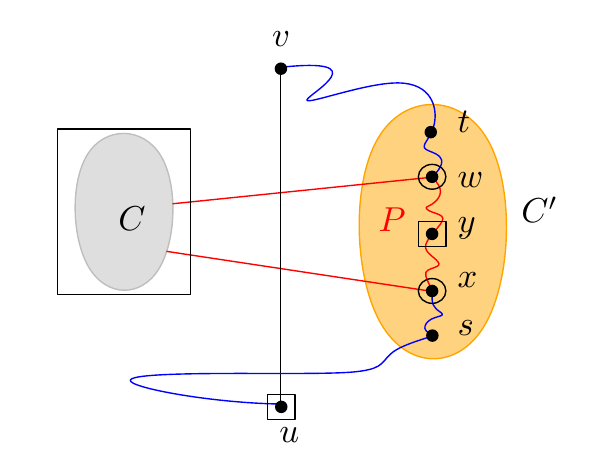}
\caption{The $K_{2,3}$-minors we use in the proof of Lemma~\ref{lem:non-connecting-component}}
\label{fig: minors_a}
\end{figure}

\begin{lemma}
 \label{lem:non-connecting-component}
Every non-connecting component $C$ of $G'-X$ is adjacent to exactly one component $C'$ of $G'-X$. Moreover, $C'$ is connecting, there are at most two vertices in $C'$ that are incident to edges in $E(C, C')$, and if there are two such vertices $w,x\in C'$, then they are adjacent and removing $wx$ disconnects $C'$.
\end{lemma}
\begin{proof}
Without loss of generality, we assume that $C$ is a left component.  
Since $C$ is non-connecting, any component adjacent to it must be connecting.  Moreover, if there are two distinct such components, they lie on the right side of the edge $uv$. 
Then either there is a path on the right side that connects them (but then they are not distinct), or removing $C$ disconnects these components, and therefore $uv$, contradicting the assumption that $C$ is a non-connecting component.
Therefore $C$ is adjacent to exactly one other component $C'$, which must be a right connecting component.

Let $w$ and $x$ be the first and the last vertex in $C'$ that are adjacent to vertices in $C$ when sweeping the vertices of $G$ clockwise in $\delta_G$ starting at $v$; see Figure~\ref{fig: minors_a}.
The lemma holds trivially if $w=x$.  
Suppose $w \neq x$. 
Next we show that $wx \in E$ and that $wx$ is a bridge of $C$.
Let $P$ be an arbitrary path from $w$ to $x$ in $C'$.  If $P$ contains an internal vertex $y$, then the path $P$ together with a path from $w$ to $x$ whose internal vertices lie in $C$ forms a cycle, where $x$ and $w$ are not consecutive.
Note that at least one of $u,v$, say $u$, is not identical to $w, x$, otherwise, $u,v$ are 2-connected.
This cycle, together with disjoint paths from $w$ to $v$ and $x$ to $u$ and the edge $uv$ yields a $K_{2,3}$-minor in $G$; see Figure~\ref{fig: minors_a}.
Such paths exist, by the outerplanarity of $\delta_{G'}$ and the fact that $C'$ is connecting, but $C$ is not.
Since $G$ is outerplanar, and therefore cannot contain a $K_{2,3}$-minor, this immediately implies that $P$ consists of the single edge $wx$, which must be a bridge of $C'$ as otherwise there would be a $wx$-path with an internal vertex. Observation~\ref{obs:no-skip-path} implies that $w$ and $x$ are the only vertices of $C$ that are adjacent to vertices in $C'$.
\end{proof}

\begin{proposition}
\label{prop:thncc}
Let $C$ be a left (right) non-connecting component of $G'-X$.  It is always possible to obtain a new almost-planar drawing $\delta'_G$ of $G$ from~$\delta_G$ by moving only the vertices of $C \setminus \{u,v\}$ to the right (left) side.
\end{proposition}
\begin{proof}
Without loss of generality, we assume that $C$ is a left component.
Since $C$ is non-connecting, then by Lemma~\ref{lem:non-connecting-component}, it is adjacent to at most two vertices on the right side. If there are two such vertices, denote them by $w$ and $x$ such that~$w$ occurs before $x$ on a clockwise traversal from~$v$ to $u$.  Note that $wx$ is a bridge of a right component $C'$ by Lemma~\ref{lem:non-connecting-component}; see Figure~\ref{fig:sect2-1_b}. 
Consider the two components of $C'-wx$ and let $y$ be the last vertex that lies in the same component as $w$ when traversing vertices clockwise from $w$ to $x$.  If~$C$ is connected to only one vertex, then we denote this by~$y$.  In both cases, let~$y'$ be the vertex of $L$ that immediately succeeds~$y$ in clockwise direction (If $y=u$, let $y'$ be the vertex that immediately precedes $y$.).

We obtain $\delta_G'$ by moving all vertices of $C \setminus \{u,v\}$ between $y$ and $y'$, reversing their clockwise ordering.  
Observe that the choice of $y$ and $y'$ guarantees that $\delta_G'$ is almost-planar and all crossings lie on $uv$.
\end{proof}
It remains to deal with connecting components.
 \begin{figure}[htb]
        \centering
\includegraphics[page=2]{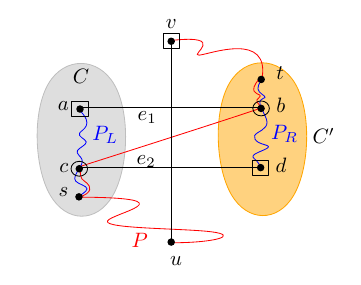}
\caption{The $K_{2,3}$-minor we use in the proof of Lemma~\ref{lem:connecting-component}}
\label{fig: minors_b}
\end{figure}

\begin{lemma}
  \label{lem:connecting-component}
The connecting 
component of $G'-X$ containing $u$ or $v$ 
is adjacent to at most one connecting component.  Every other connecting component is adjacent to exactly two connecting components. 
Moreover, if $C$ and $C'$ are two adjacent connecting components,  then there is a vertex $w$ that is incident to all edges in $E(C,C')$.
\end{lemma}
\begin{proof}
The claims concerning the adjacencies of the connecting components follows from the fact that every $uv$-path visits all connecting components in the same order.  
It remains to prove that all edges between two connecting components share a single vertex.
If $u$ and $v$ are in one component, then this component is the only connecting component and there is nothing to show.  

Now let $C$ and $C'$ be adjacent connecting components.
We assume without loss of generality, that $C$ is a left and $C'$ is a right component.
For the sake of contradiction, assume there exist two edges $e_1,e_2 \in E(C,C')$ that do not share an endpoint.  
Let $e_1 = ab$ and $e_2=cd$ where $a,c \in C$ and $b,d \in C'$ such that their clockwise order is $a,b,d,c$; see Figure~\ref{fig: minors_b}.
Note that one of $u,v$ is not in the set $\{a,b,c,d\}$.
Otherwise, $u$ and $v$ are 2-connected, which contradicts our case assumption. In the following, we assume without loss of generality that $a,b,c,d,v$ are five distinct vertices in~$G'$.   
Let $P$ be a path from $u$ to $v$ in $G'$. 
Since $C$ and $C'$ are both connecting, $P$ contains vertices from both components.  
When traversing $P$ from $u$ to $v$, let $s$ and $t$ denote the first and the last vertex of $C \cup C'$ that is encountered, respectively.
Here, we assume without loss of generality that $s \in C$ and $t \in C'$.
Let $P_L$ be a path in $C$ that connects $s$ to $a$ and let $P_R$ be a path in $C'$ that connects $d$ to $t$. 
By Observation~\ref{obs:no-skip-path}, $P_L$ contains $c$ and $P_R$ contains $b$. 
We then obtain a $K_{2,3}$-minor of $G$ by contracting each of the paths $P_L[c,a]$, $P_R[d,b]$, $vuP[u,s]P_L[s,c]$, and $P_R[b,t]P[t,v]$ into a single edge.
\end{proof}

By Lemma~\ref{lem:non-connecting-component} and Lemma~\ref{lem:connecting-component}, all vertices of a connecting component of $G'-X$ can be moved to the other side, similarly as in Proposition~\ref{prop:thncc}.
\begin{proposition}
\label{prop:thcc}
Let $C$ be a left (right) connecting component of $G'-X$.  It is possible to obtain a new almost-planar drawing $\delta'_G$ of $G$ from~$\delta_G$ by moving only the vertices of $C \setminus \{u,v\}$ to the right (left) side.
\end{proposition}
\begin{proof}
We assume without loss of generality that $C$ is a left connecting component. 
Now, we determine two vertices $w$ and $w'$ such that a right component is a non-connecting component adjacent to $C$ iff it lies between $w$ and $w'$ entirely.
If $u,v$ are not in $C$, by Lemma~\ref{lem:connecting-component}, $C$ is adjacent to exactly two right connecting components $C'$, $C''$(see~Figure~\ref{fig:one_side_move_a}). 
In the following, we assume that $v, C', C'', u$ are in clockwise order. 
Let $w$ be the last vertex in $C'$ and $w'$ be the first vertex in $C''$ when traversing the vertices in $\delta_G$ clockwise from $v$.
If $C$ contains both $u$ and $v$, let $w$ be $v$ and $w'$ be $u$. If $C$ contains either $u$ or $v$, by Lemma~\ref{lem:connecting-component}, $C$ is adjacent to exactly one right connecting components $C'$.  
Assume without loss of generality that $v \in C$. 
Let $w$ be the last vertex in $C'$ when traversing the vertices in $\delta_G$ clockwise and $w'$ be $u$. 
Observe that, due to the connectivity of $G'$ and the outerplanarity of $\delta_{G'}$, each right component that entirely lies between $w$ and let $w'$ is a non-connecting component adjacent to $C$.
Again, we want to only move the component $C$ to the right side between $w$ and $w'$ without introducing any crossings.

For simplicity, we describe the procedure in two phases.
In the first phase, we move all the right non-connecting components connected to $C$ to the left side ``temporarily'' by the procedure described in the proof of Proposition~\ref{prop:thncc} such that the components are merged into $C$ on the left; see Figure~\ref{fig:one_side_move_b}.
In the second phase, we move the component $C\setminus\{u,v\}$ (alongside the vertices that are moved in the first phase) to the right side between $w$ and $w'$, reversing their clockwise ordering; see Figure~\ref{fig:one_side_move_c}.
For each right component $C'$ that is adjacent to $C$, by Lemma~\ref{lem:connecting-component}, there is exactly one vertex shared by edges $E(C, C')$.
Thus, there is no crossing on the right side of $uv$ after the second phase. Furthermore, the vertices moved to the left at the first phase are in the same order as in $\delta_G$ after two reversals and they still lie between $w$ and $w'$.
Therefore, we can reach the same order after this two-phase procedure by only moving the vertices in $C$ to the right side accordingly. 
\end{proof}

\begin{figure}[tb]
    \centering
	    \begin{subfigure}[t]{0.33\textwidth}
	        \centering
	            \includegraphics[page=4]
		{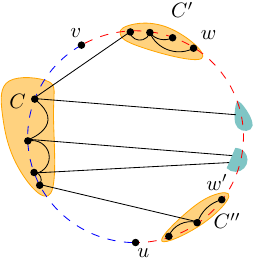}
	\caption{to move $C$ to the right side}
	\label{fig:one_side_move_a}
	    \end{subfigure}
		\hfill
    \begin{subfigure}[t]{0.33\textwidth}
        \centering
\includegraphics[page=5]
		{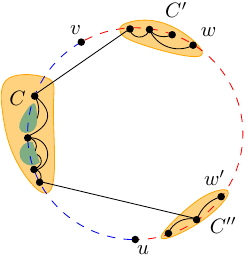}
    \caption{1.Phase}
        \label{fig:one_side_move_b}
    \end{subfigure}%
		\hfill
		\begin{subfigure}[t]{0.33\textwidth}
        \centering
\includegraphics[page=6]
		{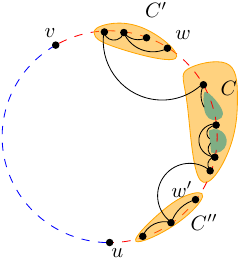}
    \caption{2.Phase}
        \label{fig:one_side_move_c}
    \end{subfigure}%
        \caption{Case 2.2 (connecting-component)}
\label{fig:one_side_move}
\end{figure}

Proposition~\ref{prop:thncc} and Proposition~\ref{prop:thcc} together imply Theorem~\ref{th:one_side_move}.

In the last part of this section, we consider optimal untangling under the restriction that the positions of $u$ and $v$ are fixed.
We denote such untangling as \textit{edge-fixed untangling}.
\begin{theorem}
\label{theoremFixShifting}
Given an almost-planar drawing $\delta_{G}$ of an outerplanar graph $G$, an edge-fixed untangling of $\delta_G$ with the minimum number of vertex moves can be computed in linear time. 
\end{theorem}
\begin{proof}
Let $e=uv$ be the crossed edge of $\delta_G$, let $C$ be a connected component of $G' = G-e$, and let $L_C = L \cap C$ and  $R_C = R \cap C$. 
Note that every edge-fixed untangling  must either move $L_C$ entirely to the right or $R_C$ entirely to the left of edge $e$.
Thus, any edge-fixed untangling must move $\min\{|L_C|, |R_C|\}$ vertices in each component $C$.

It remains to prove that we can compute such a move sequence with the minimum number of required vertex moves for each component $C$.
If $u, v$ are not connected in $G'$, the claim is exactly the same as Proposition~\ref{prop:thcc1}.
We now consider the case that $u, v$ are connected in $G'$.
Let $C$ be the connected component of $G'$ that contains both $u$ and $v$. 
We can always move either $L_C$ or $R_C$ by Proposition~\ref{prop:thncc} and Proposition~\ref{prop:thcc}.
Note that any other connected component $C'$ of $G'$ must lie entirely to the left or entirely to the right of edge $e$ since $\delta_{G'}$ is planar and $u,v$ are connected in $G'$.
\end{proof}

\section{Untangling Almost-Planar Drawings}
\label{sec: restricted untangling}
Finally, we consider how to untangle an almost-planar circular drawing $\delta_G$ of an $n$-vertex outerplanar graph $G=(V,E)$ with the minimum number of vertex moves. 
The main result of this section is the following theorem, which we prove by combining the claims of two propositions.

\begin{theorem}\label{thm:untanglealmost}
We can compute a minimum untangling for an almost-planar circular drawing $\delta_G$ of an $n$-vertex outerplanar graph $G=(V,E)$ in $O(n^2)$ time.
\end{theorem}

Let $e=uv$ be the edge of $\delta_G$ that contains all the crossings, and let $G' = G-e$ and $\delta_{G'}$ be the straight-line circular drawing of $G'$ by removing the edge $e$ from $\delta_G$. 
The edge $uv$ partitions the vertices in $V \setminus \{u,v\}$ into the sets $L$ and $R$ that lie on the left and right side of the edge $uv$ (directed from $u$ to $v$).
Let $C_u$ and $C_v$ be the connected components of $G'$ that contain $u$ and $v$, respectively. 
Note that $C_u = C_v$ if $u, v$ are in the same connected component of $G'$.


\begin{proposition}
\label{pro:NoFix}
It is always possible to untangle $\delta_G$ by moving only the vertices of $C_u$ or only the vertices of $C_v$.
\end{proposition}

\begin{proof}
If $C_u = C_v$ the claim is trivially true. 
So let us consider the case that $u$ and $v$ are not connected in $G'$. 
We describe the untangling by moving $C_u$ entirely as follows; with the same idea, we can untangle $\delta_G$ by moving $C_v$.
Let $\sigma_u$ be the clockwise order of $C_u$ in $\delta_{G'}$, starting with $u$. 
We insert the vertices of $C_u$ in the order $\sigma_u$ clockwise right after $v$ to obtain a new drawing $\delta'_{G'}$ of $G'$.
Since $C_u$ was crossing-free before and is placed consecutively on the circle, it remains crossing-free.
No other edges have been moved. 
Furthermore, $u$ and $v$ are now neighbors on the circle, so we can insert the edge $uv$ without crossings and have untangled $\delta_G$.
\end{proof}


It is clear from Proposition~\ref{pro:NoFix} that we can untangle $\delta_G$ by moving all vertices of the smaller of the two connected components $C_u$ and $C_v$, so we obtain $\shifto(\delta_G) \le \min\{|C_u|, |C_v|\}$. Assuming that the untangling from Proposition~\ref{pro:NoFix} is not minimal, we need to find a minimum untangling with $\shifto(\delta_G) < \min\{|C_u|, |C_v|\}$ vertex moves. Thus it remains to consider the case, where some vertices of $C_u$ and some vertices of $C_v$ are not moved; we call such unmoved vertices \emph{fixed} vertices and an untangling with fixed vertices in both components $C_u$ and $C_v$ a \emph{component-fixed untangling}. Then Theorem~\ref{thm:untanglealmost} is obtained by choosing the untangling with fewer vertex moves from the ones provided by Propositions~\ref{pro:NoFix} and~\ref{theoNoFixComponent}.

\begin{proposition}
\label{theoNoFixComponent}
A component-fixed untangling $U$ with the minimum number of vertex moves can be found in $O(n^2)$ time.
\end{proposition}

The reminder of this section is devoted to computing the untangling $U$. 
We distinguish between two cases based on whether $u,v$ are connected in $G'$ or not. 
In each case, we present an untangling that runs in $O(n^2)$ time and reports an optimal component-fixed untangling.

We introduce some notions and provide basic observations.
Let $G$ be a connected outerplanar graph.
Let $B$ be a $2$-connected component of $G$ and $E(B)$ the set of edges in $B$.
Since $G$ is connected and $B$ is $2$-connected, each connected component of $G-E(B)$ contains exactly one vertex in $B$.
Given a vertex $b$ in $B$, let $C_b$ be the connected component of $G-E(B)$ that contains $b$\remove{ and assume $C$ is not empty}. We denote $C_b$ as the \textit{attachment} of the $2$-connected component $B$ at the vertex $b$.

Let $H(B)$ be the cyclic vertex ordering of $B$ in the order of its Hamiltonian cycle.
Recall that every biconnected outerplanar graph has a unique Hamiltonian cycle~\cite{SYSLO197947}.
We get Observation~\ref{obsAttachmentConsecutive}; see Figure $\ref{fig:attachment}$.

\begin{observation}
\label{obsAttachmentConsecutive}
\new{Let $\delta_G$ be a planar circular drawing of an outerplanar graph $G$ and $B$ be a 2-connected component of $G$. 
Then, the clockwise cyclic vertex ordering of $B$ in $\delta_G$ is either $H(B)$ or its reverse. Furthermore, for each attachment of $B$, its vertices appear consecutively on the circle in $\delta_G$.}
\end{observation}

\begin{figure}[htbp]
	\centering
\includegraphics[page=1]{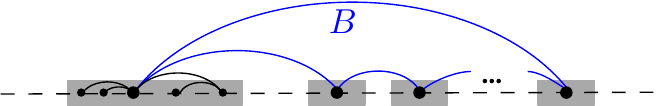}
\caption{A 2-connected component $B$ (in blue) and its attachments (gray boxes) in a planar circular drawing.
}    
\label{fig:attachment}
\end{figure}

Given a connected outerplanar graph $G$, a $2$-connected component $B$ of $G$ and a circular drawing $\delta_G$, 
we say a sequence $S$ of vertex moves \new{of $G$} is \textit{canonical}, with respect to  $B$, if in the drawing obtained by applying $S$ \new{to $\delta_G$}, the clockwise cyclic vertex ordering of each attachment of $B$ remains unchanged.
Now we are ready to show that an optimal component-fixed untangling with the restriction that fixed vertices exist in \new{both of} $C_u$ and $C_v$ can be found in 
$O(n^2)$ time; see Theorem~\ref{theoNoFixComponent}.

\medskip
\noindent\textbf{Case 1: $u$ and $v$ are connected in $G'$.} Let $C$ be a connected component of $G'$ that does not contain $u,v$.
\new{We claim now that $C$ must lie entirely on one side of $uv$ in $\delta_{G}$. 
Otherwise, let $P$ be a path of $\delta_{G'}$ that connects $u$ and $v$. Then there would exist crossings between edges of $P$ and edges of $C$ in $\delta_{G'}$ which contradicts the fact that $\delta_{G'}$ has no crossings.
Thus, we can ignore such components as they do not need to be involved in an untangling.}
Hence, we may assume $G'$ is a connected graph. 
If $u$ and $v$ are $2$-connected in $G'$, then $\delta_G$ is already outerplanar; see Proposition~\ref{prop:Biconected}.
Now we consider the case that $u$ and $v$ are connected, but not 2-connected in $G'$.
Note that $u,v$ are 2-connected in $G$. 
Let $B$ be the 2-connected component of $G$ that contains $u,v$.
We prove that each component-fixed untangling $U$ can be transformed into a canonical untangling with smaller or the same number of vertex moves; see Lemma~\ref{lemCanonical}.
Thus, we restrict our attention to canonical untanglings.
Let $H(B) = b_1,\ldots b_k$ be the cyclic vertex ordering of the Hamiltonian cycle of $B$.
Let $A_i$ be the attachment of $B$ at the vertex $b_i$ and let $\sigma(A_i)$ be the clockwise vertex ordering of $A_i$ in $\delta_G$ for $i \in \{1,\ldots, k\}$.
We consider an optimal canonical component-fixed untangling $U^o$ which orders $B$ clockwise as $H(B)$. 
Let $\delta^{''}_G$ be the outerplanar drawing obtained by applying $U^o$.
Then the clockwise vertex ordering of $\delta^{''}_G$ is exactly the concatenation of 
$\sigma(A_1),\sigma(A_2),\ldots, \sigma(A_k)$. 
Given $\delta^{''}_G$, an optimal untangling transforming $\delta_G$ to $\delta^{''}_G$ can be computed in $O(n^2)$ time; see~\cite{DBLP:journals/corr/abs-1208-0396}. 
Analogously, we obtain an optimal component-fixed untangling $U^r$ which orders $B$ counterclockwise as $H(B)$. 
From the two untanglings $U^o$ and $U^r$, we report the one which moves less vertices as the optimal component-fixed untangling.

\begin{lemma}
\label{lemCanonical}
Let $B$ be the 2-connected component of $G$ that contains $u,v$. 
Every component-fixed untangling $U$ of $\delta_G$ can be transformed into a canonical vertex move sequence $U^c$ (with respect to  $B$) that untangles $\delta_G$.
Furthermore, the number of vertex moves in $U^c$ is not greater than 
the number of vertex moves in $U$.
\end{lemma}
\begin{proof}
\new{Given a component-fixed untangling $U$ of $\delta_G$, let $\delta^U_G$ be the drawing obtained after applying $U$ on $\delta_G$. 
In $\delta^U_G$, the cyclic vertex ordering of $B$ (clockwise or counterclockwise) must correspond to its Hamiltonian cycle ordering $H(B)$. 
Furthermore, the vertices of each attachment of $B$ appear consecutively in $\delta^U_G$, including  one vertex of $B$; see Observation~\ref{obsAttachmentConsecutive}.
Let $A_1,\ldots,A_k$ be the attachments of $B$ in $G$ (indexed in clockwise order as in $\delta^U_G$) and let $\sigma(A_i)$ be the clockwise vertex ordering of $A_i$ in $\delta_G$ for $i \in \{1\ldots k\}$.
Now consider the vertex ordering $\sigma'_G=$($\sigma(A_1),\cdots,\sigma(A_k)$) and let $\delta'_G$ be an arbitrary circular drawing where the vertices are ordered as $\sigma'_G$. 
Note that the vertex ordering of each attachment is $\sigma(A_i)$ in $\delta'_G$ as in the almost-planar drawing $\delta_G$, thus each attachment in $\delta'_G$ is crossing-free. 
Moreover, in $\delta'_G$ the vertices of $B$ are ordered as in the planar drawing $\delta^U_G$, thus there is no crossing inside $B$.  
Overall, $\delta'_G$ is a planar circular drawing.
Let $U^c$ be the untangling of $\delta_G$ with minimum number of vertex moves such that the clockwise vertex ordering of the resulting drawing is $\sigma'_G$.
}

\new{
To see that $U^c$ does not move more vertices than $U$, let $\sigma_G$ and $\sigma^U_G$ be the clockwise vertex orderings of $\delta_G$ and  $\delta^U_G$, respectively.
We can observe that any common subsequence of  $\sigma_G,\sigma^U_G$ is a subsequence of $\sigma'_G$.
}
\end{proof}

\noindent\textbf{Case 2: $u$ and $v$ are not connected in $G'$.}
Note that a connected component of $G'$ that lies entirely on one side of $uv$ in $\delta_{G}$ can be ignored, since there is no need to move any vertices in such components.
After ignoring such components, we can assume that a connected component $C$ of $G'$ either contains $u,v$ or $C$ contains vertices from $L$ and also vertices from $R$.

\begin{observation}
\label{obs:cu_cv_consecutive}
In $\delta_{G'}$, vertices of $C_u$ (resp. $C_v$) lie consecutively on the cycle.
\end{observation}

The first step of our untangling $U$ deals with the connected components of $G'$ that neither contain $u$ nor $v$.
Let $U^\mathrm{fix}$ be an arbitrary component-fixed untangling of $\delta_G$, and let  $\delta^\mathrm{fix}_G$ be the outerplanar drawing of $G$ obtained from $\delta_G$ by applying $U^\mathrm{fix}$.  

\begin{lemma}
\label{lem:lower_bound_hypo}
Let $C$ be a connected component of $G'$ that does not contain vertices 
$u$ or $v$.
Let $f_u, f_v$ be two vertices in $C_u$ and $C_v$, respectively, which are fixed in $\delta^\mathrm{fix}_G$.
Then, $C$ must lie entirely on one side of $f_uf_v$\footnote{\new{Given a circular drawing of $G=(V,E)$, two vertices $a, b$ partitions the vertices in $V\setminus\{a,b\}$ into two sets that lie on the left side and right side of the ray $\overrightarrow{\rm ab}$.}   } in $\delta^\mathrm{fix}_G$.
\end{lemma}
\begin{proof}
In the graph $G$, due to the definition of $f_u$ and $f_v$, there exists a path $P_1$ in $C_u$ connecting $f_u$ to $u$, and a path $P_2$ in $C_v$ connecting $v$ to $f_v$; see Figure~\ref{fig:oneside_fixed}. Then, the path $P=P_1 uv P_2$ in $G$ connects $f_u$ to $f_v$.
In $\delta^\mathrm{fix}_G$, suppose that the connected component $C$ is not entirely on one side of $f_uf_v$, it implies that at least one edge $xy$ in $C$ has endpoints $x, y$ alternate with $f_u, f_v$ in clockwise ordering of $\delta^\mathrm{fix}_G$ and then has crossings with $P$.
It contradicts the outerplanarity of the drawing $\delta^\mathrm{fix}_G$. 
\end{proof}
\begin{figure}[htb]
	\centering
\includegraphics[page=3]{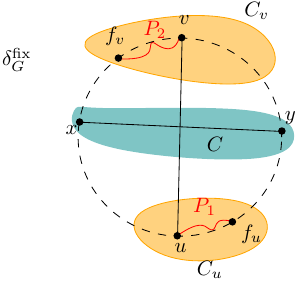}
\caption{An example illustration for the proof of Lemma~\ref{lem:lower_bound_hypo}.}
\label{fig:oneside_fixed}
\end{figure}

Now let $C$ be a connected component that does not contain $u.v$. Vertices $f_u$ and $f_v$ partition the vertices of $C$ in drawing $\delta_G$ into two sets $L_C$ and $R_C$ that are encountered clockwise and counter-clockwise from $f_u$ to $f_v$ in $\delta_G$, respectively. 
Observe that, $L_C = L\cap C$ and $R_C = R\cap C$; see Observation~\ref{obs:cu_cv_consecutive}. 
Let $m(C) = \min\{|L\cap C|,|R\cap C|\}$.
By Lemma~\ref{lem:lower_bound_hypo}, $m(C)$ is a lower bound of the moved vertices in $C$ in a component-fixed untangling.
By Proposition~\ref{prop:thcc1}, we can move $m(C)$
vertices of $C$ such that $C$ lies entirely on one side of $uv$.
In the first step of our untangling $U$, we repeat this step for each component not containing $u$ or $v$. 
After that, an almost-planar drawing of $G$ remains that has already each component not containing $u$, $v$ placed entirely on one side of $uv$. 
We can ignore such components from now on since they never need to be moved again.

Now we assume that $G'$ has exactly two connected components, namely $C_u$ and $C_v$.
Consider an arbitrary outerplanar drawing  $\delta'_G$ of $G$.
Let $\sigma(\delta'_G)$ be the circular ordering of vertices in $\delta'_G$ encountered clockwise. 
Observe that, in $\sigma(\delta'_G)$, the vertices of $C_u$ (resp. $C_v$) are in a consecutive subsequence $\sigma(C_u)$ (resp. $\sigma(C_v)$).
Otherwise, alternating vertices of two connected components would  introduce crossings.

\begin{figure}[tb]
	\centering
\includegraphics[page=2]{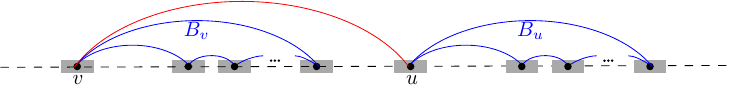}
\caption{In any clockwise vertex ordering of a planar circular drawing, $u, v$ must be the extreme vertices in the 2-connected components $B_v$ and $B_u$, respectively.}    
\label{fig:double-circle}
\end{figure}
Given an edge $e'$ in $C_v$, we say $e'$ \textit{covers} $v$ if the endpoints of $e$ alternate with $u$ and $v$ in $\delta_{G'}$.
Note that there is no edge covering $v$ in $\sigma(C_v)$. 
Otherwise, such an edge would cross with edge $uv$. 
Therefore, in a valid untangling of $\delta_G$, it is necessary to move vertices of $C_v$ in $\delta_G$ such that no crossing is introduced in $C_v$ and $v$ is not covered by any edges in $C_v$. 
Similarly, the same claim holds also for $C_u$.
We call such vertex moves \textit{vertex unwrapping}.  
In the following, we consider how to find a valid unwrapping of $v$ with the minimum number of vertex moves.  
The same operation will be also applied to $C_u$.
Observe that, once $u, v$ are both unwrapped, adding the edge $e$ into the drawing does not introduce any crossings.
The combination of these two unwrappings makes an optimal untangling. 
Here, we also consider the canonical vertex sequences and get the following 
Lemma~\ref{lemUnwrap}.
The proof is quite similar to the proof of Lemma~\ref{lemCanonical} which concerns canonical untanglings.

\begin{observation}
\label{obsUnwrap}
There exists at least one 2-connected component $B$ of $C_v$ such that $B$ contains $v$ and no edge in the attachment of $v$ (with respect to $B$) covers $v$ in $\delta_G$.
\end{observation}
The reason for this observation is that either no 2-connected component $B$ containing $v$ contains an edge covering $v$, in which case $v$ is already unwrapped and the statement is true for any such $B$. Or some 2-connected component $B$ does contain a covering edge, but then the attachment of $v$ in $B$ cannot cover $v$ due to planarity of $\delta_{G'}$.

\begin{lemma}
\label{lemUnwrap}
Let $B$ be a 2-connected component of $C_v$ that contains $v$ such that the attachment of $v$ contains no edge covering $v$. 
Each unwrapping $W$ of $v$ can be transformed into a canonical unwrapping $W^c$ (with respect to $B$).
Furthermore, the number of vertex moves in $W^c$ is not greater than 
the number of vertex moves in the original unwrapping~$W$.
\end{lemma}
\begin{proof}
\new{Given an unwrapping $W$ of $v$, let $\delta^W_G$ be the drawing obtained after applying $W$ on $\delta_G$.
In $\delta^W_G$, the cyclic vertex ordering of $B$ (clockwise or counterclockwise) must correspond to its Hamiltonian cycle ordering $H(B)$. 
Furthermore, the vertices of each attachment of $B$ appear consecutively in $\delta^W_G$, including  one vertex of $B$; see Observation~\ref{obsAttachmentConsecutive}.
Let $A_1,...A_k$ be the attachments of $B$ in $C_v$ (in this clockwise order in $\delta^W_{G}$), let $\sigma(A_i)$ be the clockwise vertex ordering of $A_i$ in $\delta_{G}$ for $i \in \{1\ldots k\}$.
Consider the clockwise vertex ordering $\sigma'_{G}$ where the vertices of $B\cup C_u$ are ordered as in 
$\delta^W_{G}$.
Furthermore, for each attachment $A_i$ the vertices of $A_i$ appear consecutively in the clockwise ordering $\sigma(A_i)$. 
Let $\delta'_G$ be an arbitrary circular drawing where the vertices are ordered as $\sigma'_G$.
Note that the vertex ordering of each attachment of $B$ is $\sigma(A_i)$ in $\delta'_G$ as in the almost-planar drawing $\delta_G$, thus each attachment in $\delta'_G$ is crossing-free. 
Moreover, in $\delta'_G$ the vertices of $B$ are ordered as in the planar drawing $\delta^W_G$, thus there is no crossing inside $B$.  
Overall, the vertex $v$ is unwrapped in $\delta'_{G}$.
It remains to prove that the canonical unwrapping $W^c$, which transforms $\delta_{G}$ to $\delta'_{G}$, moves less than or equally many vertices of $C_v$ as $W$.
By construction each common subsequence of $\delta_{G}$ and $\delta^W_G$ is also a subsequence of $\delta'_{G}$, which implies this fact.}
\end{proof}

By Lemma~\ref{lemUnwrap}, we restrict our attention to canonical unwrappings.
Fixing a 2-connected component $B_v$ of $C_v$ containing $v$ such that no edge in the attachment (with respect to $B_v$) of $v$ covers $v$, we consider the two possible canonical unwrappings of $v$, which respectively order vertices of $B$ clockwise along $H(B)$ or its reversal, and compute the corresponding resulting clockwise vertex ordering $\sigma_v$ and $\sigma^{rev}_v$ of $C_v$.
With the same idea, we get the clockwise vertex orderings $\sigma_u$ and $\sigma^{rev}_u$ of $C_u$ by the canonical unwrappings of $u$. 
We then get the four optimal unwrappings, each of them transforming $\delta_G$ to one of the vertex orderings $(\sigma_v\sigma_u), (\sigma^{rev}_v\sigma_u), (\sigma_v\sigma^{rev}_u)$ and $(\sigma^{rev}_v\sigma^{rev}_u)$. 
Such optimal unwrappings can be computed in $O(n^2)$ time; 
see~\cite{DBLP:journals/corr/abs-1208-0396}.
We report the one that moves the minimum number of vertices as an optimal component-fixed untangling.

\section{Conclusion and Outlook}
We introduced and investigated the problem of untangling non-planar circular drawings.
First from the computational side, we demonstrated the \NP-hardness of the problem \untangling.
Second, we studied the almost-planar circular drawings, where all crossings involve a single edge.  
We gave a tight upper bound of~$\lfloor\frac{n}{2}\rfloor-1$ on the shift number and an~$O(n^2)$-time algorithm to compute it.  Open problems for future work include: (i) The parameterized complexity of computing the circular shifting, e.g., with respect to the number of crossings or the number of connected components.
(ii) Generalization of our results for almost-planar drawings.
(iii) Investigation of minimum untangling by other elementary moves such as swapping vertex pairs or moving larger chunks of vertices.

\bibliographystyle{abbrv} 
\bibliography{untangling.bib}





\end{document}